\documentclass[5pt]{article}
\usepackage{amsmath}
\usepackage[latin1]{inputenc}
\usepackage{amsmath,amsthm,amssymb}
\usepackage{graphicx}
\usepackage[all,poly,knot]{xy}

\usepackage{color}

\renewcommand{\theequation}{\thesection.\arabic{equation}}

\newtheorem{lem}{Lemma}[section]
\newtheorem{thm}{Theorem} [section]
\newtheorem{prop}{Proposition} [section]

\newtheorem{coro}{Corollary}[section]

\title{Compositional Inverses of AGW-PPs
\thanks{Supported By NSF of China No. 12171163 }}

\author{Pingzhi Yuan\thanks{ P. Yuan is with School of  of Mathematical Science, South China Normal University,  Guangzhou 510631, China (email: yuanpz@scnu.edu.cn).}}

    \date{}
\begin{document}
\baselineskip15pt \maketitle
\renewcommand{\theequation}{\arabic{section}.\arabic{equation}}
\catcode`@=11 \@addtoreset{equation}{section} \catcode`@=12

    \begin{abstract}In this paper, we present two methods to obtain the compositional inverses of AGW-PPs. We improve some known results  in this topic. 

\end{abstract}


\section{Introduction}

\,\,\, Let $q$ be a prime power, $\mathbb{F}_q$ be the finite field of order $q$, and $\mathbb{F}_q[x]$
be the ring of polynomials in a single indeterminate $x$ over $\mathbb{F}_q$. A polynomial
$f \in\mathbb{F}_q[x]$ is called a {\em permutation polynomial} (PP for short) of $\mathbb{F}_q$ if it induces
a one-to-one map from $\mathbb{F}_q$ to itself.

Permutation polynomials over finite fields have been an interesting
subject of study for many years, and have applications in coding
theory \cite{LC07}, cryptography \cite{RSA, SH}, combinatorial
design theory \cite{DY06}, and other
areas of mathematics and engineering. Information about properties,
constructions, and applications of permutation polynomials may be
found in Lidl and Niederreiter \cite{LN97,LN86}, and Mullen \cite{Mull}.

In 2011, Akbrary, Ghioca and Wang  \cite{AGW11} proposed a powerful method called the AGW criterion for constructing PPs. A PP is called AGW-PP when a PP is constructed using the AGW criterion or it can be interpreted by the  AGW criterion. AGW-PPs can be divided into three types: multiplicative type, additive type and hybrid type.
 It is difficult to obtain the explicit compositional inverse of a random PP, except for several well-known classes.   Compositional
inverses of different classes of PPs of special forms have been
obtained in explicit or implicit forms; see \cite{CH02, LQW19, LN97, NLQW21, TW14, TW17, W17, Wu14, WL13, WL13J, ZYLHZ19, ZY18, ZWW20} for more details.  Recently, Niu, Li, Qu and Wang \cite{NLQW21} obtained a general method to finding compositional inverses of AGW-PPs. They obtained the compositional inverses of all AGW-PPs of the type $x^rh(x^s)$ over $\mathbb{F}_q$, where $s|q-1$.  However there are many other classes of AGW-PPs whose compositional inverses are unknown, and it is not easy to follow the framework in \cite{NLQW21} to obtain the compositional inverses of AGW-PPs for other types. The purpose of the present paper is to find other general methods to obtain the compositional inverses of  AGW-PPs.

The rest of this paper is organized as follows. In Section 2, we prove some results related to the AGW criterion. In particular, we obtain a useful commutative diagram, which is essential for the proofs of our main theorems. In Section 3, we use the dual diagram obtained in Section 2 to finding the compositional inverses of AGW-PPs. We improve some results in \cite{NLQW21}. In Section 4, we describe another method to compositional inverses of PPs by using the  branch functions.

\section{ AGW criterion and the dual diagram}

In this section, we present some  results related to the AGW criterion, and we will give the dual diagram when  the AGW criterion is applied to  a bijective function $f$.

The following lemma is taken from \cite[Lemma 1.1]{AGW11}, which is called AGW criterion now.

\begin{lem}\label{lem-1.1}   {\rm ( AGW criterion) }
Let $A, S$ and $\bar{S}$ be finite sets with $\sharp
S=\sharp\bar{S}$, and let $f:A\rightarrow A, h:
S\rightarrow\bar{S}$, $\lambda: A\rightarrow S$, and $\bar{\lambda}:
A\rightarrow \bar{S}$ be maps such that $\bar{\lambda}\circ
f=h\circ\lambda$. If both $\lambda$ and $\bar{\lambda}$ are
surjective, then the following statements are equivalent:

(i) $f$ is a bijective (a permutation of $A$); and

(ii) $h$ is a bijective from $S$ to $\bar{S}$ and if $f$ is
injective on $\lambda^{-1}(s)$ for each $s\in S$.
\end{lem}

We also have other  results for the bijection of a map. We have
\begin{lem}\label{le2.2} Let $A, S$ be finite sets,  $f: A\to A$ a map and $\lambda: A\to S$ a surjective map. Then $f$ is a bijection if and only if

(i)  $f(x)$ is
injective on each $\lambda^{-1}(s)$ for all $s \in S$.

(ii) If  $\lambda(a)\ne\lambda(b)$, then $f(a)\ne f(b)$.

Moreover (ii) is equivalent to $f\left(\lambda^{-1}(s_1)\right)\cap f\left(\lambda^{-1}(s_2)\right)=\emptyset$ for any distinct $s_1, s_2\in S$.
\end{lem}
\begin{proof} If $f$ is a bijective, then (i) holds trivially. For any distinct $s_1, s_2\in S$, since $\lambda^{-1}(s_1)\cap\lambda^{-1}(s_2)=\emptyset$, we have $f\left(\lambda^{-1}(s_1)\right)\cap f\left(\lambda^{-1}(s_2)\right)=\emptyset$, and (ii) holds.

Conversely, suppose that (i) and (ii) hold, since $f(x)$ is
injective on each $\lambda^{-1}(s)$ for all $s \in S$ and $f\left(\lambda^{-1}(s_1)\right)\cap f\left(\lambda^{-1}(s_2)\right)=\emptyset$  for any distinct $s_1, s_2\in S$, we have
$$\sharp f(A)=\sum_{s\in S} \sharp f\left(\lambda^{-1}(s)\right)=\sum_{s\in S}\sharp\lambda^{-1}(s)=\sharp A,$$
which implies that $f$ is a surjective, and thus $f$ is a bijection.
\end{proof}

\begin{lem}\label{le2.3} Let $A, S, \bar{S}$ be finite sets with $\sharp S=\sharp\bar{S}$,  $f: A\to A$ a map and $\lambda: A\to S$ a surjective map. Then $f: A\to A$ is a bijection if and only if the following two conditions hold:

(i)  $f(x)$ is
injective on each $\lambda^{-1}(s)$ for all $s \in S$.

(ii) There  exists a maps pair $(\bar{\lambda}, h)$ such that  $\bar{\lambda}: A \to \bar{S}$ is a surjective,  $h: S\to \bar{S}$ is a bijective and the following diagram commutes.
$$\xymatrix{
  A \ar[d]_{\lambda} \ar[r]^{f}
                & A \ar[d]^{\bar{\lambda}}  \\
  S  \ar[r]_{h}
                & \bar{S}            }
$$
\end{lem}

\begin{proof} (i) is obvious.
If $f: A\to A$ is a bijection, then, by the proof of Lemma \ref{le2.2} , $A$ is a disjointed union of $\lambda^{-1}(s_i), s_i\in S$, and $A$ is also a disjointed union of $f(\lambda^{-1}(s_i)), s_i\in S$, i.e.
$$A=\uplus_{s\in S}f\left(\lambda^{-1}(s)\right)=\uplus_{s\in S}\lambda^{-1}(s).$$
Now for any bijection $h: S\to \bar{S}, s\mapsto h(s)$, we define a map $\bar{\lambda}: A \to \bar{S}$ by $\bar{\lambda}(a_s)=h(s)$ for any $s\in f\left(\lambda^{-1}(s)\right)$. It is easy to check that $\bar{\lambda}: A \to \bar{S}$ is a surjective and the diagram in the theorem commutes. This proves that (ii) holds.

If (i) and (ii) holds, then by AGW criterion, $f$ is a bijection.

\end{proof}

{\bf Remark:} It is not difficult to see that there are precisely $(\sharp S)!$ maps pairs $(\bar{\lambda}, h)$ such that (ii) holds.

We also have the following result by Lemma \ref{le2.3}.

\begin{coro} Let $A, S, \bar{S}$ be finite sets with $\sharp S=\sharp\bar{S}$,  $f: A\to A$ a map,  $\lambda: A\to S$ a surjective map and $h: S\to \bar{S}$ is a bijective. Then $f: A\to A$ is a bijection if and only if the following two conditions hold:

(i)  $f(x)$ is
injective on each $\lambda^{-1}(s)$ for all $s \in S$.

(ii) There  exists a unique determined surjective map   $\bar{\lambda}: A \to \bar{S}$ such that the following diagram commutes.
$$\xymatrix{
  A \ar[d]_{\lambda} \ar[r]^{f}
                & A \ar[d]^{\bar{\lambda}}  \\
  S  \ar[r]_{h}
                & \bar{S}            }
$$
\end{coro}

\begin{lem} Let $A, S$ be finite sets,  $f: A\to A$ a map and $\lambda: A\to S$ a surjective map. Suppose that there is a set $\bar{S}$ and a map $\bar{\lambda}: A \to \bar{S}$ such that $\bar{\lambda}(f(a))= \bar{\lambda}(f(b))$ for any $a, b\in A$ with $\lambda(a)=\lambda(b)$. Then there exists a unique map $h: S\to \bar{S}$ such that the following diagram commutes.
$$\xymatrix{
  A \ar[d]_{\lambda} \ar[r]^{f}
                & A \ar[d]^{\bar{\lambda}}  \\
  S  \ar[r]_{h}
                & \bar{S}            }
$$
Furthermore, $f$ is a bijection  if the following two conditions hold

(i)  $f(x)$ is
injective on each $\lambda^{-1}(s)$ for all $s \in S$.

(ii) $h$ is an injection.

\end{lem}
\begin{proof} Since $\lambda$ is surjective, so it is easy to check that
$$h:\lambda(a)\to \bar{\lambda}(f(a))$$
satisfies $h(\lambda(a))=\bar{\lambda}(f(a))$ for all $a\in A$, so the above diagram commutes.

If (i) and (ii) hold, then for any $a, b\in A$ with $\lambda(a)\ne\lambda(b)$, then $h(\lambda(a))\ne h(\lambda(b))$, and so $\bar{\lambda}(f(a))\ne \bar{\lambda}(f(a))$, which implies that $f(a)\ne f(b)$. By Lemma \ref{le2.2}, we conclude that $f$ is a bijection.

\end{proof}
\begin{lem} Let $A, S$ be finite sets,  $f: A\to A$ a map and $\lambda: A\to S$ a surjective map. Then $f$ is a bijection if and only if the following conditions hold

(i) $f(x)$ is
injective on each $\lambda^{-1}(s)$ for all $s \in S$.

(ii) there exist a set $\bar{S}$, an injective $h: S\to \bar{S}$ and a map $\bar{\lambda}: A\to \bar{S}$ such that the following diagram commutes.$$\xymatrix{
  A \ar[d]_{\lambda} \ar[r]^{f}
                & A \ar[d]^{\bar{\lambda}}  \\
  S  \ar[r]_{h}
                & \bar{S}            }
$$\end{lem}
\begin{proof} Suppose that $f$ is a bijection, then we take
$$\bar{S}=\{[f(\lambda^{-1}(s))], s\in S\},$$
$\bar{\lambda}: f(a)\to [f(\lambda^{-1}(s))]$ for any $a\in \lambda^{-1}(s)$ and $h:s\to [f(\lambda^{-1}(s))]$. Since $f$ is a bijection, so $\bar{\lambda}$ is a well-defined surjective map. Moreover, $\bar{\lambda}(f(a))= \bar{\lambda}(f(b))=[f(\lambda^{-1}(s))]$ if $f(a)=f(b)=s$ and $\sharp(S)=\sharp(\bar{S})$.

The other direction follows from Lemma 2.2.\end{proof}

\begin{lem}  Let $A, B$ be finite sets,  $f: A\to A$ a map and $g: B\to A$ a surjective map. Then $f$ is a bijection if and only if $g(f(x))$ is a surjection.\end{lem}
\begin{proof} Obviously.\end{proof}
The following result is essential in this paper, which will be used in the sequel.
\begin{thm}\label{M}Let the notations be defined as in Lemma \ref{lem-1.1}. If $f:A\to A$ is a bijection, $f^{-1}$ and $h^{-1}$ are the compositional inverses of $f$ and $h$, respectively, then the have $\lambda\circ f^{-1}=h^{-1}\circ\bar{\lambda} $, i.e. the following diagram commutes
$$\xymatrix{
  A \ar[d]_{\bar{\lambda}} \ar[r]^{f^{-1}}
                & A \ar[d]^{\lambda}  \\
  \bar{S}  \ar[r]_{h^{-1}}
                & S           }
$$ \end{thm}
\begin{proof} By assumption, we have $\bar{\lambda}\circ f=h\circ \lambda$, hence
$$h^{-1}\circ(\bar{\lambda}\circ f)\circ f^{-1}=h^{-1}\circ (h\circ \lambda)\circ f^{-1},$$
which yields $\lambda\circ f^{-1}=h^{-1}\circ\bar{\lambda} $. This completes the proof.\end{proof}
We call the diagram in Theorem \ref{M} the dual diagram of the AGW criterion.

\section{Compositional inverses of AGW-PPs}
In this section, we present one approach to finding the compositional inverses of AGW-PPs by using the dual diagram of the AGW criterion. Our first result is as follows.
\begin{thm}\label{Thm1} Let $q$ be a prime power, and $S, \bar{S}$  subsets of $\mathbb{F}_q^\ast$ with $\sharp S=\sharp\bar{S}$.
 Let  $f:\mathbb{F}_q^\ast\rightarrow \mathbb{F}_q^\ast, g:
S\rightarrow\bar{S}$, $\lambda: \mathbb{F}_q^\ast\rightarrow S$, and $\bar{\lambda}:
\mathbb{F}_q^\ast\rightarrow \bar{S}$ be maps such that both $\lambda$ and $\bar{\lambda}$ are surjective maps and $\bar{\lambda}\circ
f=g\circ\lambda$.

Let $f_1(x)$ be a PP and  $f(x)=f_1(x)h(\lambda(x))$  a AGW-PP over $\mathbb{F}_q^\ast$, and let $f_1^{-1}(x), f^{-1}(x)$ and $g^{-1}(x)$ be the compositional inverses of $f_1(x), f(x)$ and $g(x)$, respectively. Then we have
$$f^{-1}(x)=f_1^{-1}\left(\frac{x}{h(g^{-1}(\bar{\lambda}(x)))}\right).$$\end{thm}

\begin{proof} By assumption and Theorem \ref{M}, we have the following commutative diagram
\begin{center}
 \quad\xymatrix{
  \mathbb{F}_q^\ast \ar[d]_{\lambda} \ar[r]^{f} & \mathbb{F}_q^\ast \ar[d]_{\bar{\lambda}} \ar[r]^{f^{-1}} & \mathbb{F}_q^\ast \ar[d]^{\lambda} \\
  S \ar[r]^{g} & \bar{S} \ar[r]^{g^{-1}} & S   }
\end{center}
Hence $\lambda(f^{-1}(x))=g^{-1}(\bar{\lambda}(x))$. Since $f^{-1}(x)$ is the compositional inverse
of $f(x)$, we have $f(f^{-1}(x))=x$, that is
$$f_1(f^{-1}(x))h(\lambda(f^{-1}(x)))=x.$$
It follows that $f_1(f^{-1}(x))=\frac{x}{h(\lambda(f^{-1}(x)))}=\frac{x}{h(g^{-1}(\bar{\lambda}(x)))}$, which implies that
$$f^{-1}(x)=f_1^{-1}\left(\frac{x}{h(g^{-1}(\bar{\lambda}(x)))}\right).$$
This completes the proof.\end{proof}

{\bf Remark:} In Theorem \ref{Thm1}, we use $\mathbb{F}_q^\ast$ to avoid the case of $x=0$. If we use $\mathbb{F}_q$, then we must consider the case of $x=0$ independently.

For AGW-PPs in the hybrid case, we have
\begin{lem}\label{le21} {\rm (\cite[Theorem 6.3]{AGW11})} Let $q$ be any power of the
prime number $p$, let $n$ be any positive integer, and let $S$ be any
subset of $\mathbb{F}_{q^n}$ containing $0$. Let $h, k \in \mathbb{F}_{q^n}$ be any polynomials
such that $h(0) \ne 0$ and $k(0) = 0$, and let $\lambda(x)\in\mathbb{F}_{q^n}[x]$ be
any polynomial satisfying

(1) $h(\lambda(\mathbb{F}_{q^n} ))\subseteq S$; and

(2) $\lambda(a\alpha) = k(a)\lambda(\alpha)$ for all $a\in S$ and all $a\in\mathbb{F}_{q^n}$.
Then the polynomial $f(x) = xh(\lambda(x))$ is a PP for $\mathbb{F}_{q^n}$ if and
only if $g(x) = xk(h(x))$ induces a permutation of $\lambda(\mathbb{F}_{q^n})$.\end{lem}

We apply Theorem \ref{Thm1} to $f_1(x)=x$ to obtain the following corollary, which improves Theorem 22 in \cite{NLQW21}.
\begin{coro} Let the symbols be defined as in Lemma \ref{le21}.
Let $f(x) = xh(\lambda(x))$ permute $\mathbb{F}_{q^n}$ and $g^{-1}(x)$ be the
compositional inverse of $g(x) = xk(h(x))$ over $\lambda(\mathbb{F}_{q^n})$. Then
the compositional inverse of $f(x)$ is given by
$$f^{-1}(x) =\frac{x}{h(g^{-1}(\lambda(x)))}.$$\end{coro}

The following result was discovered independently by several authors, and it can be proved by AGW criterion.

\begin{lem}\label{le10} {\rm (\cite[Theorem 2.3]{PL01} \cite[Theorem 1]{W07} \cite[Lemma 2.1]{Z09})}Let $q$ be a prime power and $f(x)=x^rh(x^s)\mathbb{F}_q[x]$, where $s=\frac{q-1}{\ell}$ and $\ell$ is an integer. Then $f(x)$ permutes $\mathbb{F}_q$ if and only if

(1) $\gcd(r, \, s)=1$ and

(2) $g(x)=x^rh(x)^s$ permutes $\mu_{\ell}$. \end{lem}
Applying Theorem \ref{Thm1} to $f_1(x)=x^r$, we obtain the following result, which is a same as in \cite[Theorem 2.3]{LQW19}.
\begin{coro} Let $f(x)=x^rh(x^s)\in\mathbb{F}_q[x]$ defined in Lemma \ref{le10} be a permutation over $\mathbb{F}_q$ and $g^{-1}(x)$ be the compositional inverse of $g(x)=x^rh(x)^s$ over $\mu_{\ell}$. If $\gcd(r, q-1)=1$ and $a$ and $b$ are two positive integers satisfying $br=1+a(q-1)$. Then the compositional inverse of $f(x)$ in $\mathbb{F}_q[x]$ is given by
$$f^{-1}(x)=x^bh(g^{-1}(x^s))^{-b}.$$
\end{coro}
\begin{proof} Since $\gcd(r, q-1)=1$ and $a$ and $b$ are two positive integers satisfying $br=1+a(q-1)$, the compositional inverse of $x^r$ is $x^b$. Applying Theorem \ref{Thm1} to $f_1(x)=x^r$, we  obtain the desired result. \end{proof}

For the compositional inverses of AGW-PPs in the additive case, we have
\begin{thm}\label{Thm2}Let $q$ be a prime power, and let $S, \bar{S}$ be subsets $\mathbb{F}_q$ with $\sharp S=\sharp\bar{S}$. Let  $f:\mathbb{F}_q\rightarrow \mathbb{F}_q, g:
S\rightarrow\bar{S}$, $\lambda: \mathbb{F}_q\rightarrow S$, and $\bar{\lambda}:
\mathbb{F}_q\rightarrow \bar{S}$ be maps such that both $\lambda$ and $\bar{\lambda}$ are surjective maps and $\bar{\lambda}\circ
f=g\circ\lambda$.

Let $f_1(x)$ be a PP and  $f(x)=f_1(x)+h(\lambda(x))$  a AGW-PP over $\mathbb{F}_q$, and let $f_1^{-1}(x), f^{-1}(x)$ and $g^{-1}(x)$ be the compositional inverses of $f(x), f_1(x)$ and $g(x)$, respectively. Then we have
$$f^{-1}(x)=f_1^{-1}\left(x-h(g^{-1}(\bar{\lambda}(x)))\right).$$\end{thm}

\begin{proof} By assumption and Theorem \ref{M}, we have the following commutative diagram
\begin{center}
 \quad
\xymatrix{
  \mathbb{F}_q \ar[d]_{\lambda} \ar[r]^{f} & \mathbb{F}_q \ar[d]_{\bar{\lambda}} \ar[r]^{f^{-1}} & \mathbb{F}_q \ar[d]^{\lambda} \\
  S \ar[r]^{g} & \bar{S} \ar[r]^{g^{-1}} & S   }\end{center}
Hence $\lambda(f^{-1}(x))=g^{-1}(\bar{\lambda}(x))$. Since $f^{-1}(x)$ is the compositional inverse
of $f(x)$, we have $f(f^{-1}(x))=x$, that is
$$f_1(f^{-1}(x))+h(\lambda(f^{-1}(x)))=x.$$
It follows that $f_1(f^{-1}(x))=x-h(\lambda(f^{-1}(x)))=x-h(g^{-1}(\bar{\lambda}(x)))$, which implies
$$f^{-1}(x)=f_1^{-1}\left(x-h(g^{-1}(\bar{\lambda}(x)))\right).$$
This completes the proof.\end{proof}

 \begin{lem}\label{le15}{\rm (\cite[Theorem 6.1]{YD11})}: Assume that $F$ is a finite
field and $S, \bar{S}$ are finite subsets of $F$ with $\sharp S =\sharp \bar{S}$ such
that the maps $\lambda : F \to S$ and $\bar{\lambda} : F \to \bar{S}$ are surjective and $\bar{\lambda}$
is additive, i.e.,
$$\bar{\lambda}(x + y) = \bar{\lambda}(x) + \bar{\lambda}(y), x, y \in F.$$
Let $g_0 : S\to F$, and $g : F \to F$ be maps such that
$$\bar{\lambda}\circ(g + g_0\circ\lambda) = g\circ\lambda,$$
$g(S) = \bar{S}$ and $\bar{\lambda}(g_0(\lambda(x))) = 0$ for every $x\in F$. Then the
map $f(x) = g(x) + g_0(\lambda(x))$ permutes $F$ if and only if $g$ permutes $F$.\end{lem}

We apply Theorem \ref{Thm2} to $f_1(x)=g(x)$, we obtain Theorem 16 in \cite{NLQW21}.
\begin{coro}{\rm (\cite[Theorem 16]{NLQW21})}: Let the symbols be defined as in Lemma \ref{le15}.
Let $f(x) = g(x) + g_0(\lambda(x))$ be a permutation over $F$ and
$g^{-1}(x)$ be the compositional inverse of $g(x)$ over $F$. Then the
compositional inverse of $f(x)$ is given by
$$f^{-1}(x) = g^{-1}\left(x - g_0(g^{-1}(\bar{\lambda}(x)))\right).$$\end{coro}

 For $S\subseteq \mathbb{F}_q$, $\gamma, b\in\mathbb{F}_q$ and a map $\lambda:\mathbb{F}_q\to\mathbb{F}_q$, $\gamma$ is called
a $b$-linear translator \cite{AGW11} of $\lambda$ with respect to $S$ if
$\lambda(x + u\gamma) = \lambda(x) + ub$ for all $x\in\mathbb{F}_q$ and $u\in S$.

\begin{lem}\label{le26} (\cite[Theorem 6.4]{AGW11}): Let $S\subseteq \mathbb{F}_q$ and $\lambda:\mathbb{F}_q\to S$
be a surjective map. Let $\gamma\in \mathbb{F}_q^\ast$ be a $b$-linear translator with
respect to $S$ for the map $\lambda$. Then for any $G\in\mathbb{F}_q[x]$ which
maps $S$ into $S$, we have that $f(x) = x + \gamma G(\lambda(x))$ is a PP
of $\mathbb{F}_q$ if and only if $g(x) = x + bG(x)$ permutes $S$.\end{lem}
Applying Theorem \ref{Thm2} to $f_1(x)=x$, we get the following corollary, which improves Theorem 27 in  \cite{NLQW21}.
\begin{coro}Let $f(x) = x + \gamma G(\lambda(x))$ defined as in
Lemma \ref{le26} be a PP on $\mathbb{F}_q$ and $g^{-1}(x)$ be the compositional
inverse of $g(x) = x + bG(x)$. Then the compositional inverse
of $f(x)$ is given by
$$f^{-1}(x)=x-\gamma G(g^{-1}(\lambda(x))).$$\end{coro}

We end this section with another AGW-PP over finite fields, which is the corrected version of Theorem 3.13 in \cite{YD14}.

\begin{lem}\label{yd14}Let $n$ be a positive integer, $f(x)\in\mathbb{F}_q[x]$ a linearized polynomial such that $\gcd(l(x), x^n-1)\ne1$, where $l(x)$ is the associated polynomial of $L(x)$. Let $a\in\mathbb{F}_{q^n}^\ast$ be a solution of $L(x)=0$ and $h(x)$ is a polynomial with $h(x)^q=h(x)$. Let $L_1(x)\in\mathbb{F}_q[x]$ be a linearized polynomial. Then for every $\delta\in\mathbb{F}_{q^n}$, the polynomial $$f(x)=ah(L(x)+\delta)+L_1(x)$$
permutes $\mathbb{F}_{q^n}$ if and only if $L_1(x)$ permutes $\mathbb{F}_{q^n}$. Moreover, if $L_1(x)$ permutes $\mathbb{F}_{q^n}$, then
$$f^{-1}(x)=L_1^{-1}\left(x-ah(L_1^{-1}(L(x))+\delta)\right).$$\end{lem}
\begin{proof}By assumption, we have $h(L(x)+\delta)\in\mathbb{F}_q, x\in\mathbb{F}_{q^n}$ since $h^q(x)=h(x)$, hence $L(ah(L(x)+\delta))=h(L(x)+\delta)L(a)=0$. It follows that the following diagram commutes
$$\xymatrix{
  \mathbb{F}_{q^n} \ar[d]_{L(x)+\delta} \ar[r]^{f}
                & \mathbb{F}_{q^n} \ar[d]^{L(x)}  \\
  S  \ar[r]_{L_1(x-\delta)}
                & \bar{S}            }
$$
where $S=\{L(x)+\delta, x\in\mathbb{F}_{q^n}\}$ and $\bar{S}=\{L(x), x\in\mathbb{F}_{q^n}\}$. By Lemma \ref{le15}, $f(x)$ permutes $\mathbb{F}_{q^n}$ if and only if $L_1(x)$ permutes $\mathbb{F}_{q^n}$. Further,
if $L_1(x)$ permutes $\mathbb{F}_{q^n}$, by Theorem \ref{Thm2} and the following diagram
\begin{center}
 \quad
\xymatrix{
  \mathbb{F}_{q^n} \ar[d]_{L(x)+\delta} \ar[r]^{f} & \mathbb{F}_{q^n} \ar[d]_{L(x)} \ar[r]^{f^{-1}} & \mathbb{F}_{q^n} \ar[d]^{L(x)+\delta} \\
  S \ar[r]^{L_1(x-\delta)} & \bar{S} \ar[r]^{L_1^{-1}(x)+\delta} & S   }\end{center}
we get
$$f^{-1}(x)=L_1^{-1}\left(x-ah(L_1^{-1}(L(x))+\delta)\right).$$
\end{proof}

\section{Branch PPs and their compositional inverses}

For a permutation polynomial $f$ by AGW criterion, we have
$$\xymatrix{
  A \ar[d]_{\lambda} \ar[r]^{f}
                & A \ar[d]^{\bar{\lambda}}  \\
  S  \ar[r]_{h}
                & \bar{S}            }
$$
As $\lambda$ is a surjective map, we have that $A$ is a disjointed union of $\lambda^{-1}(s_i), s_i\in S$, i.e.
$$A=\uplus_{s\in S}\lambda^{-1}(s).$$
Let $|S|=m, S=\{s_1, \ldots, s_m\}$,
$$A_i= \lambda^{-1}(s_i),\quad B_i=f\left(\lambda^{-1}(s_i)\right), \,\, i=1, \ldots, m,$$
$$f|_{A_i}=f_i: A_i\to f(A_i)=B_i, \,\, i=1, \ldots, m.$$
Then $f: A\to A$ can be viewed as a branch function of $f_i(x),\,\, i=1, \ldots, m$. Let $f_i^{-1}: B_i\to A_i, \,\, i=1, \ldots, m,$ be the local inverse of $f_i, \,\, i=1, \ldots, m$, i.e. $f_i^{-1}(f_i(x))=x, x\in A_i, \,\, i=1, \ldots, m$. For any non-empty subset $T$ of $S$, let $\chi_T(x)$ be the characteristic function on $T$, i.e.
\begin{eqnarray*}
\chi_T(x)= \left\{ \begin{array}{ll}
                     1, & x\in T \\
                     0, & \mbox{ otherwise.}
                    \end{array}
         \right.
\end{eqnarray*}

Then we have

\begin{thm}\label{Thmb}Let $A$ be a non-empty subset of $\mathbb{F}_q$, and $f$ is a branch bijection defined as above, then we have
$$f^{-1}(x)=\sum_{i=1}^m\chi_{B_i}(x)f_i^{-1}(x).$$\end{thm}
\begin{proof} For $x\in A_i, i=1, \ldots, m$, we have $f(x)=f_i(x)\in B_i$, hence
$$(\sum_{i=1}^m\chi_{B_i}f_i^{-1})\circ f(x)=\sum_{i=1}^m\chi_{B_i}(f(x))f_i^{-1}(f(x))=f_i^{-1}(f_i(x))=x,$$ which implies that $$f^{-1}(x)=\sum_{i=1}^m\chi_{B_i}(x)f_i^{-1}(x).$$ This completes the proof.\end{proof}

For any branch bijection, it is not easy to give the characteristic functions $\chi_{B_i}(x)$. However, it is easy for the branch functions when we use the cyclotomic cosets as branches.

Let $\gamma$ be a fixed primitive element of $\mathbb{F}_q$, $s|q-1$. The integer $\ell=\frac{q-1}{s}$ is called the index of $f(x)=x^rh(x^s)$. Let $C_0$ be the set of all non-zero $\ell$-th powers, i.e. $C_0=<\gamma^\ell>$. $C_0$ is a subgroup of $\mathbb{F}_q^\ast$ of index $\ell$. The cosets of $C_0$ are the cyclotomic cosets
$$C_i: =\gamma^iC_0, \quad i=0, 1, \ldots, \ell-1. $$
Let $\mu_\ell$ denote the set of $\ell$-th roots of unity in $\mathbb{F}_q^\ast$, i.e.
$$\mu_\ell=\{x\in\mathbb{F}_q^\ast|x^\ell=1\}.$$
We have
\begin{lem}Let $\gamma$ be a fixed primitive element of $\mathbb{F}_q$, $s|q-1$ and $\ell=\frac{q-1}{s}$. Let $C_0=<\gamma^\ell>$, and
$$C_i: =\gamma^iC_0, \quad i=0, 1, \ldots, \ell-1. $$
Let $h_i(x)=1+\left(\frac{x}{\gamma^i}\right)^s+\cdots+\left(\frac{x}{\gamma^i}\right)^{(\ell-1)s}, i=0, 1, \ldots, \ell-1$. Then we have
\begin{eqnarray*}
\chi_{C_i}(x)=\frac{h_i(x)}{\ell} =\left\{ \begin{array}{ll}
                     1, & x\in C_i \\
                     0, & \mbox{ otherwise}
                    \end{array}
         \right.
\end{eqnarray*}
for $i=0, 1, \ldots, \ell-1$.\end{lem}
\begin{proof} The result follows from the following well-known result that
\begin{eqnarray*}1+g^s+\cdots+g^{(\ell-1)s}= \left\{ \begin{array}{ll}
                    \ell,  & g\in C_0 \\
                     0, & \mbox{ otherwise.}
                    \end{array}
         \right.
\end{eqnarray*} \end{proof}

 Now we will the branch function to give two proofs of Theorem 11 in \cite{NLQW21}.

\begin{prop}{\rm (\cite[Theorem 11]{NLQW21})} Let $f(x)=x^rh(x^s)\in\mathbb{F}_q[x]$ defined in Lemma \ref{le10} be a permutation over $\mathbb{F}_q$ and $g^{-1}(x)$ be the compositional inverse of $g(x)=x^rh(x)^s$ over $\mu_{\ell}$. Suppose $a$ and $b$ are two integers satisfying $as+br=1$. Then the compositional inverse of $f(x)$ in $\mathbb{F}_q[x]$ is given by
$$f^{-1}(x)=g^{-1}(x^s)^ax^bh(g^{-1}(x^s))^{-b}.$$
\end{prop}
\begin{proof} {\bf The first proof:} View $f(x)$ as a branch function of $f_i(x)=x^rh(\gamma^{is}), x\in C_i=\gamma^i<\gamma^{\ell}>, \,\, i=0, \ldots, \ell-1$. Since $h(x^s)=h(\gamma^{is})$ is a constant for any $x\in C_i$, it is easy to check that $f^{-1}_i(x)=\gamma^{ais}h(\gamma^{is})^{-b}x^b, i=0, \ldots, \ell-1$ satisfy $f_i^{-1}(f_i(x))=x, x\in C_i, i=0, \ldots, \ell-1$. Here $f_i: C_i\to C_{t(i)}, f(C_i)=C_{t(i)}$. Since $f(\gamma^{i+dt})^s=g(\gamma^{(i+dt)s})=g(\gamma^{is})$, we have
$$f(C_i)^s=g(C_i^s)=C_{t(i)}^s, \quad g(\gamma^{is})=\gamma^{t(i)s},$$
and
$$f_i^{-1}: C_{t(i)}\to C_i, \quad x\mapsto \gamma^{ias}h(\gamma^{is})^{-b}x^b, x\in C_{t(i)}, \, i=0, \ldots, \ell-1.$$
Hence $f^{-1}(x)=x^bh'(x^s)$, where $h'(\gamma^{t(i)s})=\gamma^{ias}h(\gamma^{is})^{-b}$. Since $g$ is a bijection, we have $g^{-1}(\gamma^{t(i)s})=\gamma^{is}$ and
$\gamma^{ias}h(\gamma^{is})^{-b}=(g^{-1}(\gamma^{t(i)s}))^a(h(g^{-1}(\gamma^{t(i)s})))^b$, which implies that
$$h'(x^s)=(g^{-1}(x^s))^a(h(g^{-1}(x^s)))^{-b}.$$
Hence $f^{-1}(x)=g^{-1}(x^s)^ax^bh(g^{-1}(x^s))^{-b}$, and we are done.\end{proof}

{\bf Second proof:} From the first proof we have $f^{-1}(x)=x^bh'(x^s)$, so it suffices to prove that $h'(x^s)=g^{-1}(x^s)^ax^bh(g^{-1}(x^s))^{-b}$. We use the following commutative diagram to prove this.
$$\xymatrix{
  \mathbb{F}_q^\ast \ar[d]_{x^s} \ar[r]^{f^{-1}}
                & \mathbb{F}_q^\ast  \ar[d]^{x^s}  \\
  \mu_{\ell}  \ar[r]_{g^{-1}}
                & \mu_{\ell}            }
$$
By the diagram, we have
\begin{equation}\label{eq1}g^{-1}(x^s)=x^{sb}h'(x^s)^s, \quad sa+rb=1.\end{equation}
Since $f(f^{-1}(x))=x$, we obtain
$$(x^bh'(x^s))^rh(x^{sb}h'(x^s)^s)=x, $$
i.e. $(x^bh'(x^s))^rh(g^{-1}(x^s))=x$, hence
$$x^b=(x^bh'(x^s))^{rb}h(g^{-1}(x^s))^b.$$
By (\ref{eq1}), we have
$$(g^{-1}(x^s))^a=x^{sba}h'(x^s)^{sa}=x^bh'(x^s)\left(x^bh'(x^s)\right)^{-rb}=h'(x^s)h(g^{-1}(x^s))^b.$$
It follows that $h'(x^s)=(g^{-1}(x^s))^a(h(g^{-1}(x^s)))^{-b}$, and we are done.

We end the paper with two results on branch functions.

\begin{lem}\label{B2}Let $q$ be an odd prime power, $\gamma$  a fixed primitive element of $\mathbb{F}_q$, $a_1, a_2\in \mathbb{F}_q^\ast$. Let $r_1$ and $ r_2$ be two positive integers  and let
$$f(x)=\frac{a_1}{2}x^{r_1}\left(1-x^{\frac{q-1}{2}}\right)+\frac{a_2}{2}x^{r_2}\left(1+x^{\frac{q-1}{2}}\right)$$ be defined as a branch function, i.e.
 \begin{eqnarray*}
f(x)= \left\{ \begin{array}{ll}
                     0 & x=0, \\
                     a_1x^{r_1} & x\in\gamma<\gamma^2>,\\
                     a_2x^{r_2} & x\in <\gamma^2>.
                    \end{array}
         \right.
\end{eqnarray*}
Then $f(x)$ is a PP over $\mathbb{F}_q$ if and only if $\gcd(r_1r_2, \frac{q-1}{2})=1$ and $\{(-1)^{r_1}a_1^{\frac{q-1}{2}}, a_2^{\frac{q-1}{2}}\}=\{-1, 1\}$. \end{lem}
\begin{proof}Since
\begin{eqnarray*}
x^{\frac{q-1}{2}}\circ f(x)= \left\{ \begin{array}{ll}
                     a_1^{\frac{q-1}{2}}x^{\frac{r_1(q-1)}{2}} & x\in\gamma<\gamma^2>,\\
                     a_2^{\frac{q-1}{2}}x^{\frac{r_2(q-1)}{2}} & x\in <\gamma^2>,
                    \end{array}
         \right.
\end{eqnarray*}
thus we have the following commutative diagram
$$\xymatrix{
  \mathbb{F}_q^\ast \ar[d]_{x^{\frac{q-1}{2}}} \ar[r]^{f}
                & \mathbb{F}_q^\ast  \ar[d]^{x^{\frac{q-1}{2}}}  \\
  \{-1, 1\}  \ar[r]_{h}
                & \{-1, 1\}            }
$$
where $\left(h(-1), \, h(1)\right)=\left((-1)^{r_1}a_1^{\frac{q-1}{2}}, a_2^{\frac{q-1}{2}}\right)$. Observe that $a_1x^{r_1}$ permutes $\gamma<\gamma^2>$ if and only if $\gcd(r_1, \frac{q-1}{2})=1$, and $a_2x^{r_2}$ permutes $<\gamma^2>$ if and only if $\gcd(r_1, \frac{q-1}{2})=1$. Therefore by AGW criterion $f(x)$ is a PP over $\mathbb{F}_q$ if and only if $\gcd(r_1r_2, \frac{q-1}{2})=1$ and $\{(-1)^{r_1}a_1^{\frac{q-1}{2}}, a_2^{\frac{q-1}{2}}\}=\{-1, 1\}$. This completes the proof.
\end{proof}

\begin{coro}\label{TB2} Let $f(x)$ be a AGW-PP defined as in Lemma \ref{B2}. We have

(i) If $r_1$ is odd, then $f(x)$ is a PP if and only if $\gcd(r_1r_2, \frac{q-1}{2})=1$
 and $a_1a_2$ is a square.
(ii) If $r_1$ is even, then $f(x)$ is a PP if and only if $\gcd(r_1r_2, \frac{q-1}{2})=1$
 and $a_1a_2$ is  not a square.
 Moreover, the number of such PPs is $(q-1)^2\varphi^2\left(\frac{q-1}{2}\right)/2$, where $\varphi(x)$ is the Euler function.\end{coro}

 \begin{proof} By Lemma \ref{B2}, the proofs of (i) and (ii) are obvious. For the number of such PPs, it is easy to see that $a_1x^{r_1}, a_1'x^{r_1'}, 1\le r_1, r_1'\le q-1, a_1, a_1'\in\mathbb{F}_q^\ast$ are distinct bijections on $\gamma<\gamma^2>$ if and only if $\gcd(r_1r_1', q-1)=1$ and $r_1\not\equiv r_1'\pmod{(q-1)/2}$. Therefore the number of such bijections on
 $\gamma<\gamma^2>$ is $(q-1)\varphi\left(\frac{q-1}{2}\right)$. The same result holds for the bijections on $<\gamma^2>$. Hence, by (i) and (ii),  the total number of such PPs is $(q-1)^2\varphi^2\left(\frac{q-1}{2}\right)/2$.\end{proof}

 {\bf Remark:} In fact, all such PPs in Corollary \ref{TB2} form a group with the operation of  composition of functions. We denote this group as $G_2$, then $\sharp(G_2)=(q-1)^2\varphi^2\left(\frac{q-1}{2}\right)/2$, and all functions of the form $ax^r, a\in \mathbb{F}_q^\ast, \gcd(r, q-1)=1, 1\le r\le q-1$ form a subgroup of $G_2$.  Usually, $G_2$ is not an abelian group, and thus it is not easy to give all elements of order 2 (the involutions). We will discuss  this question in another paper.

\end{document}